\documentclass[12pt]{iopart}

\expandafter\let\csname equation*\endcsname\relax
\expandafter\let\csname endequation*\endcsname\relax

\usepackage{amsmath}
\usepackage{amsfonts}
\usepackage{amssymb}
\usepackage{amsthm}
\usepackage{color}
\usepackage{bbm,dsfont}
\usepackage{graphicx}
\usepackage{hyperref}
\usepackage{enumerate}
 
\newtheorem{proposition}{Proposition}
\newtheorem{proposition?}{Proposition?}
\newtheorem{theorem}{Theorem}

\theoremstyle{definition}

\newtheorem{definition}{Definition}

\newcommand{\nat}{\mathbb N} 

\newcommand{\hi}{\mathcal{H}} 
\newcommand{\lh}{\mathcal{L(H)}} 
\newcommand{\ket}[1]{|#1\rangle} 
\newcommand{\bra}[1]{\langle#1|} 
\renewcommand{\tr}[1]{\textrm{tr}\left[#1\right]} 
\newcommand{\id}{\mathbbm{1}} 

\newcommand{\meo}{\Omega} 
\newcommand{\salg}{\mathcal{B}(\Omega)} 
\newcommand{\bor}[1]{\mathcal{B}(#1)} 

\newcommand{\A}{\mathsf{A}}
\newcommand{\B}{\mathsf{B}}
\renewcommand{\P}{\mathsf{P}}

\newcommand{\J}{\mathcal{J}}
\newcommand{\I}{\mathcal{I}}
\newcommand{\sat}[1]{\mathfrak{s}(#1)} 

\newcommand{\pleq}{\preceq}
\newcommand{\psim}{\simeq}

\begin{document}

\title[]{Saturation of repeated quantum measurements}

\author{
Erkka Haapasalo$^1$,
Teiko Heinosaari$^2$,
Yui Kuramochi$^1$
}
\address{$^1$Department of Nuclear Engineering, Kyoto University,
Kyoto daigaku-katsura, Nishikyo-ku, Kyoto, 615-8540, Japan}
\address{$^2$Turku Centre for Quantum Physics, Department of Physics and Astronomy, University of Turku, FI-20014 Turku, Finland}

\begin{abstract}
We study sequential measurement scenarios where the system is repeatedly subjected to the same measurement process. We first provide examples of such repeated measurements where further repetitions of the measurement do not increase our knowledge on the system after some finite number of measurement steps. We also prove, however, that repeating the L\"uders measurement of an unsharp two-outcome observable never saturates in this sense, and we characterize the observable measured in the limit of infinitely many repetitions. Our result implies that a repeated measurement can be used to correct the inherent noise of an unsharp observable.
\end{abstract}

\section{Introduction}\label{sec:intro}

A non-trivial quantum measurement necessarily perturbs the initial state of the measured system. Hence, subsequent measurements on the same system are typically disturbed compared to the situation without the first measurement. Despite the unavoidable measurement disturbance, we can still hope to learn more about the initial state of the system by performing measurements in sequence. For instance, by measuring sequentially unsharp versions of conjugated observables, we can implement any covariant phase space observable, hence also an informationally complete observable \cite{CaHeTo11},\cite{CaHeTo12}. Properties of sequential measurements can also be used to study quantum foundations. For instance, the demand for the existence of a sequential product on an effect algebra excludes certain types of unphysical effect algebras \cite{GuGr2002},\cite{GhGu2004}.

Given that a sequential measurement of two different observables can give more information than the independent measurements of the same observables, then how useful can a repetition of the same measurement on the same system be?
A repetition of a von Neumann measurement just gives the same outcome every time and, since all the potential of the measurement is thus used already after the first step, there is no reason to measure such an observable more than once. However, quantum observables generally also allow less invasive measurements, so we may ask how useful repeating the same measurement can be and after how many repetitions the full potential of the measurement setup is reached.

In the present work, we study repeated applications of the same measurement device and compare different numbers of repetition. Such a scheme of quantum non-demolishion measurements has been proposed in a quantum optical setting \cite{Haroche_et_al_1992} and has also been performed experimentally \cite{Haroche_et_al_2007}. Repeated measurements and interactions between the system and a chain of probes has also been studied, e.g.,\ in \cite{MePe_2015} and \cite{BruJoMe_2010} with an emphasis on asymptotic behaviour of the system state in the latter reference. To properly understand the amount of classical information retrievable from the system through a repeated measurement, we define the {\it saturation step} to be the least number of repetitions of the measurement after which we obtain no added information. We discuss cases when a single application of the measurement gives all the information that can be extracted by repetition, but also cases when the saturation step is finite but greater than one. 

Using the techniques developed recently in \cite{Kuramochi15a}, we give a precise meaning for a measurement that is repeated infinite number of times.
Our most surprising finding is that, in some repeated measurement setups, saturation is never reached after finitely many steps but, in the limit of infinitely many repetitions, we may correct the inherent noise of the original observable produced by a single application of the measurement process. We show that this is the case for the L\"uders measurement of a binary observable (apart from trivial cases), and we find that the infinite consecutive L\"uders measurement is equivalent with the minimal sharp observable associated with it. 

\section{Preliminaries}

A measurement can be described at different levels. If we are only interested in measurement outcome probabilities, description at the level of observables suffices. If we, in addition, discuss conditional state transformation caused by the measurement, then the relevant concept is that of an instrument. We start by recalling the mathematical definitions of these two concepts. For more details, we refer to \cite{PSAQT82}, \cite{OQP97}, \cite{MLQT12}, \cite{Ozawa14}. 

\emph{Notation}. In the rest of the paper $\hi$ is a fixed Hilbert space. The dimension of $\hi$ can be either finite or countably infinite. We denote by $\lh$ the set of all bounded operators on $\hi$. 

\subsection{Observables and post-processing}

A quantum observable with finite number of outcomes is described by a map $\omega\mapsto \A(\omega)$ from a finite set $\Omega_\A$ to $\lh$ such that each $\A(\omega)$ is a positive operator and $\sum_\omega \A(\omega)=\id$, where $\id$ is the identity operator on $\hi$. An observable $\A$ is called \emph{sharp} if each operator $\A(\omega)$ is a projection.

Since we need to compare two observables, we recall the following \emph{post-processing preorder} of observables \cite{MaMu90a,Heinonen05,BuDaKePeWe05}. Given two observables $\A$ and $\B$, we denote $\A\pleq\B$ if there exists a map $\kappa(\cdot | \cdot):\Omega_\A \times \Omega_\B \to [0,1]$ such that 
\begin{equation}\label{eq:markov}
\sum_{\omega} \kappa(\omega | \omega') =1
\end{equation}
for all $\omega' \in \Omega_\B$ and
\begin{equation}\label{eq:ApleqB}
\A(\omega)=\sum_{\omega'} \kappa(\omega | \omega') \B(\omega') 
\end{equation}
for all $\omega \in \Omega_\A$. Further, we denote $\A\psim\B$ if $\A\pleq\B\pleq\A$, and $\A\prec\B$ if $\A\pleq\B$ but not $\B\pleq\A$.

A map $\kappa(\cdot | \cdot):\Omega_\A \times \Omega_\B \to [0,1]$ satisfying \eqref{eq:markov} is called a \emph{Markov kernel}. Note that Markov kernels defined on a fixed product space $\Omega_\A\times\Omega_\B$ form a convex set. A special class of Markov kernels are related to \emph{relabelings}, where we relabel the measurement outcomes, possibly giving the same label for different outcomes and hence merging them, but doing nothing else. If we start from an outcome space $\Omega_\B$, a relabeling is determined by a function $f:\Omega_\B \to \Omega_\A$. The corresponding Markov kernel $\kappa_f$ is then
\begin{align}
\kappa_f(\omega | \omega') = \delta_{\omega,f(\omega')} \, ,
\end{align}
and the post-processed observable is $\A=\B\circ f^{-1}$, i.e.,
\begin{equation}\label{eq:relabeling}
\A(\omega)=\sum_{\omega': f(\omega')=\omega} \B(\omega') \, .
\end{equation}

The post-processing preorder $\pleq$ can be used to compare the information-yielding power of observables: If $\A\pleq\B$ with $\A$ defined by $\B$ and a Markov kernel $\kappa$ as in \eqref{eq:ApleqB}, measuring $\B$ gives us at least as much information on the system as measuring $\A$ does. Indeed, if we measure $\B$, we obtain the outcome statistics of $\A$, i.e.,\ the probabilities $\Pi^\A_\varrho(\omega)=\tr{\rho\A(\omega)}$ of detecting the value $\omega$ through classical data processing represented by the kernel $\kappa$ from the outcome statistics of $\B$ {\it independent of the system state $\varrho$} since
\begin{equation}
\Pi^\A_\varrho(\omega)=\sum_{\omega'}\kappa(\omega|\omega')\Pi^\B_\varrho(\omega').
\end{equation}
Thus we may say that, if $\A\pleq\B$, then $\B$ is at least as informative as $\A$. Naturally, we interpret $\A\psim\B$ to mean that $\A$ and $\B$ are informationally equivalent; we may measure either of them and obtain the same amount of classical information.

\subsection{Instruments and their composition}

A quantum instrument, defined in the Heisenberg picture, is a function $\omega\mapsto \I_\omega$ such that each $\I_\omega$ is a normal completely positive map on $\lh$ and $\sum_\omega \I_\omega(\id) = \id$. The probability of getting a measurement outcome $\omega$ in an initial state $\varrho$ is $\tr{\varrho \I_\omega(\id)}$, hence the observable $\A^\I$ determined by $\I$ is given as 
\begin{equation}
\A^\I(\omega) = \I_\omega(\id) \, .
\end{equation}

An instrument $\I$ describes not only measurement outcome probabilities but also the conditional state transformation caused by the measurement process. 
As we use the Heisenberg picture, this is reflected in the transformation of observables. If we measure first $\I$ and then some other observable $\B$ on the same system, we obtain measurement outcomes $\omega$ and $\omega'$ with the probability $\tr{\varrho \I_\omega(\B(\omega'))}$.
 
If we perform subsequently more than two measurements, then we need to form a composition of instruments \cite{DaLe70},\cite{BuCaLa90}. Suppose we have two instruments $\I$ and $\J$. We denote by $\I_{\omega}\circ\J_{\omega'}$ the functional composition of maps $\I_{\omega}$ and $\J_{\omega'}$, i.e., 
\begin{equation}
\I_{\omega} \circ \J_{\omega'} (T) = \I_{\omega}( \J_{\omega'} (T))  \, .
\end{equation}
The order of the instruments in the composition $\I_{\omega}\circ\J_{\omega'}$ is such that the instrument $\I$ has been applied to the input state first. 

\section{Repeated measurement}
\begin{center}
\begin{figure}
\includegraphics[scale=0.22]{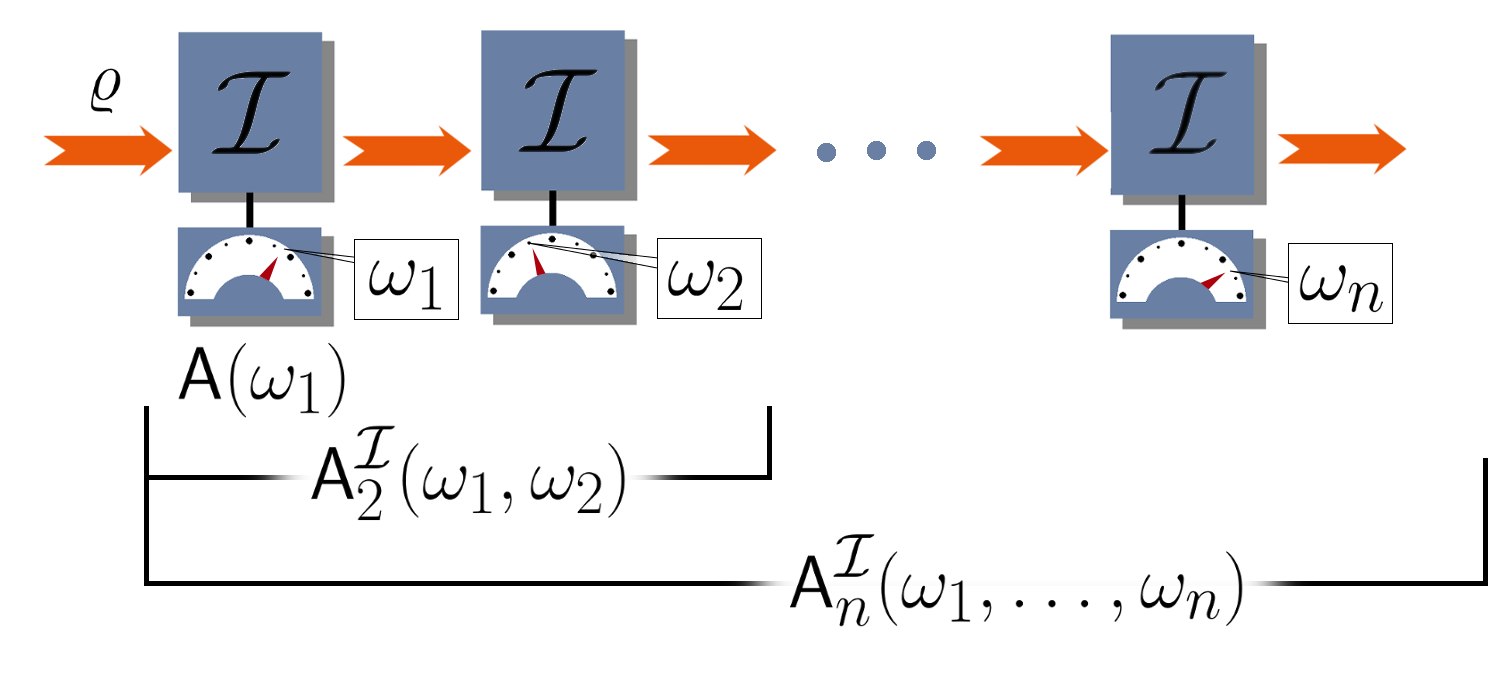}
\caption{\label{fig:chain}The measurement outcome is a finite sequence $(\omega_1 , \dots , \omega_n)$ and its probability is determined by the observable $\mathsf{A}^{\I}_n$ and the initial state $\varrho$.}
\end{figure}
\end{center}
Let $\I$ be the instrument describing a measurement device. We will now consider a repeated measurement scheme where we use the same measurement device repeatedly on the same system (Fig.~\ref{fig:chain}). If we repeat the measurement $n$ times, then the measurement outcome of the whole procedure is an element from the Cartesian product $\Omega^n$, and the probability of obtaining a particular $n$-tuple $(\omega_1,\ldots,\omega_n)$ is
\begin{equation}
\tr{\I_{\omega_1} \circ \cdots \circ \I_{\omega_n} (\id) \varrho } \, .
\end{equation}
We conclude that the observable $\A^\I_n$ corresponding to the $n$ repetitions of $\I$ is given as
\begin{equation}\label{eq:An}
\A^\I_n(\omega_1,\ldots,\omega_n)=\I_{\omega_1} \circ \cdots \circ \I_{\omega_n} (\id) \, ,
\end{equation}
and if we still make one more measurement round the related observable $\A^\I_{n+1}$ is given by the formula
\begin{align*}
\A^\I_{n+1}(\omega_1,\ldots,\omega_{n+1}) & =\I_{\omega_1} \circ \cdots \circ \I_{\omega_{n+1}} (\id) \\
& =\I_{\omega_1} (\A^\I_n (\omega_2,\ldots,\omega_{n+1}) ) \, .
\end{align*}
Since $\sum_\omega \I_\omega (\id) = \id$, we obtain
\begin{equation}
\sum_{\omega_{n+1}} \A^\I_{n+1}(\omega_1,\ldots,\omega_{n+1}) = \A^\I_{n}(\omega_1,\ldots,\omega_{n}) \, . 
\end{equation}
Thus, according to \eqref{eq:relabeling}, $\A_n$ is a relabeling of $\A_{n+1}$, so that $\A_{n} \pleq \A_{n+1}$. Therefore, we have the sequence
\begin{align}
\A^\I_1 \pleq \A^\I_2 \pleq \A^\I_3 \pleq \cdots \, ,
\end{align}
and hence by repeating the measurement we may learn more about the initial state of the system. 
The properties of this sequence will be the focus of our investigation. We start with a simple but important observation.

\begin{proposition}\label{prop:An}
If  $\A^\I_n\psim\A^\I_{n+1}$ for some $n\in\nat$, then $\A^\I_n \psim \A^\I_m$ for all $m\geq n$.
\end{proposition}

\begin{proof}
Let us assume that the claim holds for each $m=n,\ldots,\,k$ for some $k\in\nat$, $k>n$. Hence, there is a Markov kernel $\kappa(\cdot|\cdot):\Omega_\A^k\times\Omega_\A^{k-1}\to[0,1]$ such that
$$
\A^\I_k(\omega_1,\ldots,\omega_k)=\sum_{\omega'_1,\ldots,\omega'_{k-1}}\kappa(\omega_1,\ldots,\omega_k|\omega'_1,\ldots,\omega'_{k-1})\A^\I_{k-1}(\omega'_1,\ldots,\omega'_{k-1}).
$$
We may write
\begin{eqnarray*}
&&\A^\I_{k+1}(\omega_1,\ldots,\omega_{k+1})\\
&=&\I_{\omega_1}\big(\A^\I_k(\omega_2,\ldots,\omega_{k+1})\big)\\
&=&\sum_{\omega'_1,\ldots,\omega'_{k-1}}\kappa(\omega_2,\ldots,\omega_{k+1}|\omega'_1,\ldots,\omega'_{k-1})\I_{\omega_1}\big(\A^\I_{k-1}(\omega'_1,\ldots,\omega'_{k-1})\big)\\
&=&\sum_{\omega'_1,\ldots,\omega'_{k-1}}\kappa(\omega_2,\ldots,\omega_{k+1}|\omega'_1,\ldots,\omega'_{k-1})\A^\I_k(\omega_1,\omega'_1,\ldots,\omega'_{k-1})\\
&=&\sum_{\omega'_1,\ldots,\omega'_k}\tilde\kappa(\omega_1,\ldots,\omega_{k+1}|\omega'_1,\ldots,\omega'_k)\A^\I_k(\omega'_1,\ldots,\omega'_k),
\end{eqnarray*}
where
$$
\tilde\kappa(\omega_1,\ldots,\omega_{k+1}|\omega'_1,\ldots,\omega'_k)=\delta_{\omega_1,\omega'_1}\kappa(\omega_2,\ldots,\omega_{k+1}|\omega'_2,\ldots,\omega'_k)
$$
is a Markov kernel. Hence also $\A^\I_{k+1}\pleq\A^\I_k$ and, since the converse always holds, we have $\A^\I_{k+1}\psim\A^\I_k\psim\cdots\psim\A^\I_n$. Thus the claim holds for all $m$.
\end{proof}

The previous result makes the following notion meaningful.

\begin{definition}
Let $\I$ be an instrument and let $\{\A^\I_j\}_{j=1}^\infty$ be the sequence of observables defined as
\begin{equation}
\A^\I_n(\omega_1,\ldots,\omega_n)=\I_{\omega_1} \circ \cdots \circ \I_{\omega_n} (\id) \, .
\end{equation}
We say that the \emph{saturation step} of $\I$, denoted by $\sat{\I}$, is the smallest positive integer $n$ such that $\A^\I_n \psim \A^\I_{n+1}$. If $\A^\I_n \prec \A^\I_{n+1}$ for all $n\in\nat$, then we denote $\sat{\I}=\infty$.
\end{definition}

We say that an instrument $\I$ \emph{saturates at step $n$} if $n=\sat{\I}<\infty$.

\section{Instruments that saturate at the first step}

\subsection{Repeatable instrument}

An instrument $\I$ is called \emph{repeatable} if in the second repetition step we get the same measurement outcome as in the first step, with certainty \cite{DaLe70}, \cite{Ozawa84}. It is quite clear from this definition that a repeatable instrument cannot give any further information in additional repetitions. We will next show this using our framework.

Repeatability means that the probability of getting two different outcomes in subsequent measurements is zero, so for $\omega_1 \neq \omega_2$ we must have
\begin{equation}
\tr{\varrho \I_{\omega_1} \circ \I_{\omega_2} (\id)} = 0
\end{equation}
in all initial states $\varrho$.
It follows that 
\begin{equation}
\I_{\omega_1} \circ \I_{\omega_2} (\id) = \delta_{\omega_1 \omega_2} \I_{\omega_1} (\id) \, .
\end{equation}
In terms of the corresponding observables we have
\begin{equation}
\A^\I_2(\omega_1,\omega_2) = \delta_{\omega_1 \omega_2} \A^\I(\omega_1) = \sum_{\omega} \delta_{\omega \omega_1} \delta_{\omega_1 \omega_2} \A^\I(\omega) \, ,
\end{equation}
and hence $\A^\I_2 \psim \A^\I$. In conclusion, \emph{if an instrument $\I$ is repeatable, then $\sat{\I}=1$}.

\subsection{Preparative instrument}

Suppose we perform a measurement of an observable $\A$. After the measurement we prepare a state $\eta_{\omega}$ depending on the outcome $\omega$. The corresponding instrument is given by
$$
\I_\omega(T)=\tr{\eta_\omega T}\A(\omega),\qquad \omega\in\Omega_\A,\quad T\in\lh \, , 
$$
and we say that $\I$ is a preparative instrument. We have
\begin{align*}
\A^\I_2(\omega_1,\omega_2) & = \I_{\omega_1} \circ \I_{\omega_2} (\id)  = \tr{\eta_{\omega_1} \A(\omega_2)} \A(\omega_1) \\
& = \sum_\omega \delta_{\omega \omega_1} \tr{\eta_{\omega} \A(\omega_2)} \A(\omega) \, ,
\end{align*}
and hence $\A^\I_2 \psim \A$.
Therefore, $\sat{\I}=1$.

\subsection{Mixtures}\label{sec:mix}

Let $\I$ and $\J$ be two instruments related to an observable $\A$. For each $0\leq t \leq 1$, we form their convex mixture $t\I + (1-t)\J$ by defining
\begin{equation}
(t\I + (1-t)\J)_\omega (T) = t\I_\omega (T) + (1-t)\J_\omega (T)
\end{equation}
for all $\omega \in\Omega_\A$ and $T\in\lh$. Since
\begin{equation}
(t\I + (1-t)\J)_\omega (\id) = t\A(\omega ) + (1-t)\A(\omega ) = \A(\omega ) \, ,  
\end{equation}
this instrument determines the same observable $\A$ as well.

\begin{proposition}
If $\sat{\I}=\sat{\J}=1$, then $\sat{t\I + (1-t)\J}=1$.
\end{proposition}

\begin{proof}
Assume that $\sat{\I}=\sat{\J}=1$, hence $\A^\I_2 \psim \A$ and $\A^\J_2 \psim \A$. It follows that there are Markov kernels $\kappa$ and $\kappa'$ such that
\begin{equation}
\A^\I_2(\omega_1,\omega_2)=\sum_{\omega} \kappa(\omega_1,\omega_2 | \omega) \A(\omega) 
\end{equation}
and
\begin{equation}
\A^\J_2(\omega_1,\omega_2)=\sum_{\omega} \kappa'(\omega_1,\omega_2 | \omega) \A(\omega) \, .
\end{equation}
We have
\begin{align*}
\A^{t\I + (1-t)\J}_2(\omega_1,\omega_2) & = (t\I + (1-t)\J)_{\omega_1}(\A(\omega_2)) \\
& = t \A^\I_2(\omega_1,\omega_2) + (1-t) \A^\J_2(\omega_1,\omega_2) \\
& = \sum_{\omega} (t\kappa(\omega_1,\omega_2 | \omega) + (1-t)\kappa'(\omega_1,\omega_2 | \omega)) \A(\omega) \, , 
\end{align*}
thus $\A^{t\I + (1-t)\J}_2 \pleq \A$. Therefore, $\sat{t\I + (1-t)\J}=1$.
\end{proof}
 
\section{Instruments that saturate at a finite step $n>1$}

Let $\hi$ be a Hilbert space with a finite dimension $d \geq 2$ and let $\{ \phi_k  \}_{k=1}^d$ be an orthonormal basis of $\hi .$ We define an instrument $\I$ with an outcome set $\Omega = \{0,1 \}$ by
\begin{equation}
	\I_\omega (T)
	:=
	L_\omega^\ast T L_\omega \, , 
\end{equation}
where
\begin{equation}
L_0 := \ket{\phi_d}  \bra{\phi_d} \, , \quad L_1 := \sum_{k=1}^{d-1} \ket{\phi_{k+1}} \bra{\phi_k} \, .
\end{equation}
Then we have the following:

\begin{proposition}
\label{thm:ladder}
\begin{enumerate}[(a)]
\item
If $1 \leq n \leq d - 1$, then $\A^\I_n \psim \P_n$, where $\P_n$ is the sharp observable defined by
\[
	\P_n (j)
	:=
	\begin{cases}
		\displaystyle
		\sum_{l=1}^{d-n} 
		\ket{\phi_l} \bra{\phi_l} ,
		&
		(j=1)
		; \\
		\ket{\phi_{d-n+j -1}} \bra{\phi_{d-n+j-1}},
		&
		(2 \leq j \leq n+1) .
		\end{cases}
\]
\item
If $ d -1 \leq n$, then $\A^\I_n \psim \P_{d-1}$, where $\P_{d-1}$ is the sharp observable defined by
\[
	\P_{d-1} (j) := \ket{\phi_j} \bra{\phi_j}
	\quad
	(1 \leq j \leq n) .
\]
\end{enumerate}
\end{proposition}

\begin{proof}
\begin{enumerate}[(a)]
\item
Assume that $1 \leq n \leq d - 1$. Operators of the form $L_{\omega_n} L_{\omega_{n-1}} \cdots L_{\omega_1}$ are calculated to be
\begin{gather*}
	L_1^n
	=
	\sum_{l=1}^{d-n}
	\ket{\phi_{l+n}}
	\bra{\phi_l} ,
	\\
	L_0^{j} L_1^{n-j} 
	=
	\ket{\phi_{d-n+j}}
	\bra{\phi_{d-n+j}}
	\quad
	(1 \leq j \leq n) ,
\end{gather*}
and the others vanish since  $L_1 L_0  = 0 .$ Thus the nonzero operators $\A_n (\omega_1 , \cdots , \omega_n) $ are given by
\begin{gather*}
	\A_n
	(1 , \cdots , 1)
	=
	\sum_{l=1}^{d-n}
	\ket{\phi_{l}}
	\bra{\phi_l} ,
	\\
	\A_n
	(
	\underbrace{1,\cdots , 1}_{\text{$n-j$~elements}} 
	, 
	\underbrace{0 , \cdots , 0}_{\text{$j$~elements}}
	)
	=
	\ket{\phi_{d-n+j}}
	\bra{\phi_{d-n+j}} ,
	\quad
	(1 \leq j \leq n).
\end{gather*}
Hence we obtain $\A_n \psim \P_n .$
\item
Let us assume $d -1 \leq n .$ It is sufficient to show that $\A_{d-1} \psim \A_d .$
Since $\A_{d-1} \psim \P_{d-1} ,$ $\A_d$ is equivalent to an observable given by
\begin{equation*}
	\I_{\omega} (\P_{d-1} (j))
	\quad
	(\omega \in \{ 0,1 \} , 1 \leq j \leq d) .
\end{equation*}
The non-vanishing elements of this observable are given by
\begin{gather*}
	\I_{0} (\P_{d-1}  (d)  )
	=
	\P_{d-1} (d) ,
	\\
	\I_{1} (\P_{d-1} (j) )
	=
	\P_{d-1} 
	(j-1)
	\quad
	(2 \leq j \leq d) .
\end{gather*}
Thus we obtain $\A_d \psim \P_{d-1} \psim \A_{d-1} ,$ which completes the proof.\qedhere
\end{enumerate}
\end{proof}

Looking at the form of the observables $\P_j$, it is clear that 
\begin{equation*}
	\P_1 \prec \P_2 
	\prec \cdots \prec
	\P_{d-2} \prec \P_{d-1} \, .
\end{equation*}
Hence, from Prop. \ref{thm:ladder} we conclude that 
\begin{equation*}
	\A^\I_1 \prec \A^\I_2 
	\prec \cdots \prec
	\A^\I_{d-2} \prec \A^\I_{d-1} \psim \A^\I_d \psim \cdots
\end{equation*}
and therefore $\sat{\I} = d-1 .$

\section{Instruments with $\sat{\I}=\infty$}

\subsection{Observable $\A^\I_\infty$}
In the following we consider instruments with infinite saturation steps. For such cases it is convenient to consider the infinite consecutive measurement of $\I $ for which we have to study observables whose outcomes are elements of more general measurable spaces. The generalized definition of the observable is as follows. Let $(\meo , \salg)$ be a measurable space modelling the physical outcome space. A mapping $\A \colon \salg \to \lh$ is an observable, or a positive-operator-valued measure (POVM), if $\A (E) $ is positive for any $E \in \salg ,$ $\A (\Omega) = \id ,$
and $\A( \cup_j E_j  ) = \sum_j \A(E_j)$ (in the weak operator topology) for any disjoint sequence $\{ E_j \}_{j=1}^\infty \subset \salg .$ The triple $(\meo, \salg, \A)$ is also called an observable. For simplicity we will assume that the outcome spaces of observables are standard Borel spaces.

The preorder relation by classical post-processing is generalized as follows \cite{DdG97}: Let $(\meo_\A , \bor{\meo_\A} , \A)$ and $(\meo_\B , \bor{\meo_\B} , \B)$ be observables. We denote $\A \pleq \B$ if there exists a map $\kappa (\cdot | \cdot ) \colon \bor{\meo_\A } \times \meo_\B \to [0,1] $ such that $\kappa(\cdot | \omega')$ is a probability measure on $(\meo_\A , \bor{\meo_\A})$ for each $\omega' \in \meo_\B ,$ $\kappa (E | \cdot)$ is $\bor{\meo_\B}$-measurable for each $E \in \bor{\meo_\A} ,$ and 
\begin{equation*}
	\A (E)
	=
	\int_{\meo_2}
	\kappa(E|\omega_2)
	d \B (\omega_2)
\end{equation*}
for each $E \in \bor{\meo_1}$. As before, we denote $\A \psim \B$ if $\A \pleq \B $ and $\B \pleq \A$, and $\A \prec \B$  if $\A \pleq \B $ but not $\B \pleq \A$. The map $\kappa(\cdot | \cdot)$ is called a Markov kernel.

Let $\Omega$ be a finite set and let $\I$ be an instrument on $\Omega$. We define an observable $\A^\I_\infty$ with the countable Cartesian product space $(\meo^\infty , \bor{\meo^\infty}),$ called the \textit{infinite composition} of $\I ,$ by
\begin{equation}\label{eq:Ainf}
	\A^\I_\infty 
	(\{ \omega_1 \} \times \cdots \times \{ \omega_n \} \times \meo^\infty )
	=
	\A^\I_n
	(\omega_1 , \dots , \omega_n ) 
\end{equation}
for each $n\geq 1$ and $(\omega_1 , \dots , \omega_n ) \in \meo^n ,$ where $\A^\I_n$ is given by~\eqref{eq:An} and $\bor{\meo^\infty}$ is the $\sigma$-algebra generated by the cylinder sets $\{ \omega_1 \} \times \cdots \times \{ \omega_n \} \times \meo^\infty .$ The existence and uniqueness of $\A^\I_\infty$ satisfying \eqref{eq:Ainf} can be proved~\cite{Kuramochi15a} by using a POVM version of the Kolmogorov extension theorem \cite{Tumulka08}. Physically speaking, the observable $\A^\I_\infty$ corresponds to the infinite consecutive measurement of $\I .$

From the definition of the infinite composition~\eqref{eq:Ainf}, we see that $\A^\I_n \pleq \A^\I_\infty$ for any $n \geq 1.$ Furthermore, we have the following:
\begin{theorem}
\label{prop:step_Ainf}
Let $\I$ be an instrument with a finite outcome space $\meo ,$ and let $\A^\I_n$ $(n\geq 1)$ and $\A^\I_\infty$ be observables given by~\eqref{eq:An} and \eqref{eq:Ainf}, respectively. Then we have
\begin{equation}\label{eq:sat}
	\sat{\I}
	=
	\inf
	\{
	n \in \nat
	|
	\A^\I_n \psim \A^\I_\infty 
	\} .
\end{equation}
\end{theorem}

\begin{proof}
We denote the RHS of~\eqref{eq:sat} as $\sat{\I}^\prime .$ From 
\begin{equation}
\A^\I_{\sat{\I}^\prime} \pleq \A^\I_{\sat{\I}^\prime  +1 } \pleq \A^\I_\infty \psim \A^\I_{\sat{\I}^\prime} 
\end{equation}
we have $\A^\I_{\sat{\I}^\prime}  \psim \A^\I_{\sat{\I}^\prime +1} $ and thus $\sat{\I} \leq \sat{\I}^\prime . $ In order to show the converse inequality, it is sufficient to prove $\A^\I_\infty \pleq \A^\I_{\sat{\I}}$. Without loss of generality, we may assume $\sat{\I} < \infty .$ Since the observable corresponding to a consecutive measurement of $\I$ followed by $\A^\I_{\sat{\I}}$ is equivalent to $\A^\I_{\sat{\I}}$ itself, from Theorem~4 of \cite{Kuramochi15a}, we obtain $\A^\I_\infty \pleq \A^\I_{\sat{\I}} ,$ and the assertion holds.
\end{proof}

\subsection{L\"uders instrument of a two-outcome observable}\label{subsec:luders}
Let $A\in\lh$ be an effect, i.e. $0 \leq A \leq \id ,$ and let $\I$ be an instrument
with the outcome space $\Omega = \{0,1\}$ and defined by
\begin{equation}\label{eq:LudersInstrument}
\I_0 (T) = \sqrt{\id -A} T \sqrt{\id - A},\qquad\I_1 (T) = \sqrt{A} T \sqrt{A}
\end{equation}
for all $T\in\lh$. Let $(\meo^\infty , \bor{\meo^\infty})$ be the countable product space of $\meo ,$ and let $\A^\I_n$ and $\A^\I_\infty$ be the observables for the consecutive measurement processes defined by \eqref{eq:An} and \eqref{eq:Ainf}, respectively. Let $([0,1] , \bor{[0,1]} , \P^A )$ be the spectral measure of $A$ such that
\begin{equation*}
	A 
	=
	\int_{[0,1]} 
	\lambda 
	d \P^A (\lambda) ,
\end{equation*}
where $\bor{[0,1]}$ is the $\sigma$-algebra generated by the intervals. Then we have the following:
\begin{theorem}
\label{prop:PA}
\begin{equation}\label{eq:PA}
	\A^\I_\infty \psim \P^A .
\end{equation}
\end{theorem}

\begin{proof}
We first show $\A^\I_\infty \pleq \P^A .$ We define a Markov kernel $\kappa_1 (\cdot | \cdot) \colon 2^\Omega\times[0,1] \to [0,1] $ by
\[
	\kappa_1 (\{ 1 \} |\lambda  ) = \lambda\, ,\quad  \kappa_1 (\{ 0 \} |\lambda )
	=
	1-\lambda,\qquad0\leq\lambda\leq1,
\]
and we also define a probability measure $\kappa_\infty (\cdot | \lambda)$ on $(\meo^\infty , \bor{\meo^\infty}$ by the infinite direct product of $\kappa_1 (\cdot | \lambda) .$ Since $\kappa_1 (\{ \omega_1 \}| \cdot)$ is measurable, the Dynkin class theorem assures that $[0,1] \ni \lambda \mapsto \kappa_\infty (E|\lambda)$ is measurable for each $E \in \bor{\meo^\infty},$ and therefore $\kappa_\infty(\cdot | \cdot)$ is a Markov kernel. For each $n \geq 1$ and each $(\omega_i)_{i=1}^n \in \Omega^n , $ we have
\begin{align*}
	\A^\I_\infty (\{ \omega_1 \} \times \cdots \times \{ \omega_n \} \times \Omega^\infty)
	&=
	\I_{\omega_1} \circ \cdots \circ \I_{\omega_n}
	(\id)
	\\
	&=
	A^{
	\sum_{i=1}^n \omega_i
	} (\id  - A)^{n - \sum_{i=1}^n \omega_i}
	\\
	&=
	\int_{[0,1]}
	\lambda^{\sum_{i=1}^n \omega_i} 
	(1- \lambda )^{n - \sum_{i=1}^n \omega_i}
	d \P^A ( \lambda ) 
	\\
	&=
	\int_{[0,1]}
	\kappa_\infty
	(\{ \omega_1 \} \times \cdots \times \{ \omega_n \} \times \Omega^\infty | \lambda  )
	d \P^A (\lambda ),
\end{align*}
implying
$$
\A^\I_\infty(E)=\int_{[0,1]}\kappa_\infty(E|\lambda)d\P^A(\lambda)
$$
for all $E\in\bor{\Omega^\infty}$. Hence $\A^\I_\infty \pleq \P^A .$

In order to show $\P^A \pleq \A^\I_\infty ,$ we define a sequence of stochastic variables $X_n \colon \Omega^\infty \to [0,1]$ by
\begin{equation*}
	X_n (\omega) 
	:=
	\frac{1}{n}
	\sum_{i=1}^n \omega_i ,
\end{equation*}
where $\omega := (\omega_i)_{i=1}^\infty .$ The strong law of large numbers implies that $X_n (\omega) $ converges to $\lambda$ $\kappa_\infty (\cdot | \lambda)$-almost surely for every $\lambda \in [0,1].$ If we define a stochastic variable $X_\infty \colon \Omega^\infty \to [0,1]$ by
\begin{equation*}
	X_{\infty} (\omega)
	=
	\begin{cases}
		\displaystyle
		\lim_{n \to \infty}
		X_n (\omega)
		&\text{if the limit exists;}
		\\
		0 
		&\text{otherwise,} 
	\end{cases}
\end{equation*}
then $X_{\infty} (\omega) = \lambda$ $\kappa_\infty (\cdot | \lambda)$-almost surely.
Thus for each Borel subset $E$ of $[0,1],$ we have
\begin{align*}
	\int_{\Omega^\infty}
	\chi_E (X_\infty (\omega))
	d \A^\I_\infty (\omega)
	&=
	\A^\I_\infty
	(  X_\infty^{-1}  (E)   )
	\\
	&=
	\int_{[0,1]}
	\kappa_\infty
	( X_\infty^{-1}  (E)   |   \lambda   )
	d \P^A ( \lambda)
	\\
	&=
	\int_{[0,1]}
	\left(
	\int_{\meo^\infty}
	\chi_E (X_\infty (\omega))
	\kappa_\infty
	( d \omega  |   \lambda   )
	\right)
	d \P^A (\lambda)
	\\
	&=
	\int_{[0,1]}
	\left(
	\int_{\meo^\infty}
	\chi_E (\lambda)
	\kappa_\infty
	( d \omega  |   \lambda   )
	\right)
	d\P^A ( \lambda)
	\\
	&=
	\int_{[0,1]}
	\chi_E (\lambda)
	d \P^A ( \lambda)
	\\
	&=
	\P^A(E),
\end{align*}
where $\chi_E (\cdot)$ is the indicator function. The above equation implies the desired relation $\P^A \pleq \A^\I_\infty ,$ which completes the proof.
\end{proof}
The next theorem states that the L\"uders instrument of \eqref{eq:LudersInstrument} does not saturate after any finite number of steps except for in a couple of trivial cases.

\begin{theorem}
\label{thm:specA}
Let $A$ be an effect and $\I$ the associated L\"uders instrument given by \eqref{eq:LudersInstrument}.
\begin{enumerate}[(a)]
\item \label{item:specA1}
If $A$ is a projection, then $\sat{\I}  = 1$.
\item \label{item:specA2}
If $A = \lambda \id$ for some $0\leq\lambda\leq 1$, then $\sat{\I}=1$.
\item \label{item:specA3}
In all other cases $\sat{\I} = \infty .$
\end{enumerate}
\end{theorem}
\begin{proof}
The assertion \eqref{item:specA1} follows from the repeatability of the instrument $\I .$ When $A = \lambda \id ,$ we have $\A^\I_\infty \psim  \P^A \psim \{  \id \} $
and the assertion \eqref{item:specA2} is obvious.

Now we prove the assertion \eqref{item:specA3}, for which it is sufficient to show $\A^\I_n \prec \P^A$ for each $n \geq 1 .$ By the assumption, there exist $\lambda_1 $ and $\lambda_2$ in the spectrum $\sigma (A)$ of $A$ such that
\[
	0 < \lambda_1 <1,
	\quad
	\lambda_1 \neq \lambda_2 .
\]
We fix $\epsilon > 0$ such that $\epsilon < \lambda_1  < 1 - \epsilon$ and
\begin{gather*}
	[\lambda_1 - \epsilon , \lambda_1 + \epsilon]
	\cap
	[\lambda_2 - \epsilon , \lambda_2 + \epsilon]
	=
	\emptyset \, .
\end{gather*}
Since $\lambda_1 , \lambda_2 \in \sigma (A) ,$ the operators
\[
	P_1 := \P^A (  [\lambda_1 - \epsilon , \lambda_1 + \epsilon]   ),
	\quad
	P_2 := \P^A(  [\lambda_2 - \epsilon , \lambda_2 + \epsilon]   )
\]
are non-zero orthogonal projections. Hence, there exist normalized vectors $\psi_1$ and $\psi_2$ such that $P_i \psi_i = \psi_i $ $(i=1,2) .$ We will show that
\begin{equation}
	H^2 (\Pi_{\psi_1}^{\A^\I_n} , \Pi_{\psi_2}^{\A^\I_n}  )
	<
	H^2
	(\Pi_{\psi_1}^{\P^A} , \Pi_{\psi_2}^{\P^A}  )
	=1,
	\label{eq:hellinger}
\end{equation}
where $H^2 (\cdot, \cdot)$ is the square of the Hellinger distance defined by
\[
	H^2( p , q)
	:=
	\frac{1}{2}
	\int_{\meo^\prime}
	\left(
	\sqrt{d p(\omega^\prime)}   -   \sqrt{dq(\omega^\prime)}    
	\right)^2
\]
for arbitrary probability measures $p$ and $q$ on a measurable space $(\meo^\prime ,\bor{\meo^\prime}) $, and
\begin{gather*}
	\Pi_{\psi}^{\A^\I_n}
	(\omega)
	:=
	\bra{\psi }  \A^\I_n (\omega)  \ket{\psi}    ,
	\quad
	\Pi_{\psi}^{\P^A}
	(E)
	:=
	\bra{\psi } \P^A (E)  \ket{\psi    } ,
\end{gather*}
are the distributions of the corresponding measurement outcomes when the system is prepared in vector state $\psi .$ Since $\A^\I_n \psim \P^A $ implies 
$
	H^2 (\Pi_{\psi}^{\A^\I_n} , \Pi_{\psi^\prime}^{\A^\I_n}  )
	=
	H^2
	(\Pi_{\psi}^{\P^A} , \Pi_{\psi^\prime}^{\P^A}  )
$
for all unit vectors $\psi$ and $\psi^\prime$~\cite{JePuVi08}, the claim immediately follows from \eqref{eq:hellinger}.

Now we prove \eqref{eq:hellinger}. Since $\Pi_{\psi_1}^{\P^A}$ and $\Pi_{\psi_2}^{\P^A}$ are concentrated on disjoint intervals $[\lambda_1 - \epsilon , \lambda_1 + \epsilon ] $ and $[\lambda_2 - \epsilon , \lambda_2 + \epsilon ]  ,$ respectively, they are singular to each other and hence $H^2 (\Pi_{\psi_1}^{\P^A} , \Pi_{\psi_2}^{\P^A}  ) =1.$ On the other hand, for each $\omega = (\omega_i )_{i=1}^n \in \Omega^n $ we have
\begin{align*}
	\Pi_{\psi_1}^{\A^\I_n} (\omega)
	&=
	\int_{[0,1]}
	\lambda^{\sum_{i=1}^n \omega_i   } (1-\lambda)^{n- \sum_{i=1}^n \omega_i  }
	d \Pi^{\P^A}_{\psi_1} (\lambda)
	\\
	&=
	\int_{[\lambda_1 - \epsilon,  \lambda_1 + \epsilon  ]}
	\lambda^{\sum_{i=1}^n \omega_i   } (1-\lambda)^{n- \sum_{i=1}^n \omega_i  }
	d \Pi^{\P^A}_{\psi_1} (\lambda)
	\neq 0 .
\end{align*}
Therefore,
\begin{align*}
	H^2  (\Pi_{\psi_1}^{\A^\I_n} , \Pi_{\psi_2}^{\A^\I_n}  )
	&=
	\frac{1}{2}
	\sum_{\omega \in \meo^n}
	\left(
	\sqrt{  \Pi_{\psi_1}^{\A^\I_n} (\omega) }   -   \sqrt{ \Pi_{\psi_2}^{\A^\I_n}  (\omega)  }    
	\right)^2
	\\
	&=
	1 -
	\sum_{\omega \in \meo^n}
	\sqrt{ \Pi_{\psi_1}^{\A^\I_n} (\omega) \Pi_{\psi_2}^{\A^\I_n} (\omega)     } < 1.
\end{align*}
Thus we obtain~\eqref{eq:hellinger}, which completes the proof.
\end{proof}

\section{Discussion}

In our investigation we have given examples of instruments $\I$ with $\sat{\I}=1,\,2,\ldots$ (saturation after finitely many steps) and with $\sat{\I}=\infty$ (insaturable measurement). The natural question is to characterize these types of instruments. For instance, we do not yet know if the mixtures of repeatable and preparative instruments are the only instruments with the saturation step $1$. Moreover, it is still unsure how the saturation step of a mixture of instruments relates to the saturation steps of the component instruments in the convex decomposition.

We have shown that repeated measurements can be used to cancel the noise in an unsharp two-outcome observable. It remains to be seen if similar noise reduction occurs in other measurements than those corresponding to the L\"uders instruments. We hope that this work stimulates further theoretical and experimental investigations into the full potential of repeated uses of a single quantum instrument. 

\section*{Acknowledgements}

The support of Japan Society for the Promotion of Science is acknowledged by E.H. (as an Overseas Researcher under Postdoctoral Fellowship of Japan Society for the Promotion of Science) and by Y.K. (KAKENHI Grant No. 269905).

\section*{References}

\end{document}